\documentclass[11pt]{article}

\usepackage{fullpage,amsmath,xspace,amssymb,epsfig,amsthm,bbm,hyperref}

\newtheorem{Thm}{Theorem}
\newtheorem{Lem}[Thm]{Lemma}
\newtheorem{Cor}[Thm]{Corollary}
\newtheorem{Prop}[Thm]{Proposition}

\usepackage{algorithm}
\usepackage{algorithmic}
\usepackage{paralist}

\newcommand\mbR{\mbox{$\mathbb{R}$}}

\newcommand\D{\mbox{\sf {D}}\xspace}

\newcommand\Q{\mbox{\sf {Q}}\xspace}

\newcommand\qxor{\mbox{\sf QuantumXOR}\xspace}
\newcommand\ket[1]{| #1 \rangle}

\newcommand\gf{\mbox{$\mathbb{F}_2$}\xspace}
\newcommand\gfn{\mbox{$\mathbb{F}_2^n$}\xspace}

\newcommand\B{\{0,1\}}     
\newcommand\pmB{\{+1,-1\}}     
\newcommand\Bn{\{0,1\}^n}
\newcommand\BntB{\{0,1\}^n\rightarrow \{0,1\}}
\newcommand\BntR{\{0,1\}^n\rightarrow \mbR}

\newcommand {\st} {\textit{s.t.}\xspace}

\newcommand\pr{\mbox{\bf Pr}}
\newcommand\av{\mbox{\bf{\bf E}}}

\newcommand\alice{\mbox{\sf {Alice}}\xspace}
\newcommand\bob{\mbox{\sf {Bob}}\xspace}
\newcommand\sign{\mbox{\tt {sign}}\xspace}
\newcommand\rank{\mbox{\tt {rank}}\xspace}

\newcommand\supp{\mbox{\tt {supp}}\xspace}

\newcommand\ans{\mbox{\tt {ans}}\xspace}
\newcommand\fn[2]{\| \hat{#1} \|_{#2}}
\newcommand\wfn[2]{\| \widehat{#1} \|_{#2}}

\newtheorem*{Thm-bounded-error}{Theorem~\ref{thm:bounded-error}}
\newtheorem*{Thm-zero-error}{Theorem~\ref{thm:zero-error}}

\begin{document}
\title{\bf Efficient quantum protocols for XOR functions}
\author{Shengyu Zhang}
\date{}
\maketitle

\begin{abstract}
	We show that for any Boolean function $f:\B^n \to \B$, the bounded-error quantum communication complexity $\Q_\epsilon(f\circ \oplus)$ of XOR functions $f(x\oplus y)$ satisfies that $\Q_\epsilon(f\circ \oplus) = O\big(2^d \big(\log\fn{f}{1,\epsilon} +\log \frac{n}{\epsilon}\big)\log(1/\epsilon)\big)$, where $d = \deg_2(f)$ is the \gf-degree of $f$, and $\fn{f}{1,\epsilon} = \min_{g:\|f-g\|_\infty \leq \epsilon} \fn{g}{1}$. This implies that the previous lower bound $\Q_\epsilon(f\circ \oplus) = \Omega(\log\|\hat f\|_{1,\epsilon})$ by Lee and Shraibman \cite{LS09} is tight for $f$ with low \gf-degree. The result also confirms the quantum version of the Log-rank Conjecture for low-degree XOR functions. In addition, we show that the exact quantum communication complexity satisfies $\Q_E(f) = O(2^d \log\fn{f}{0})$, where $\fn{f}{0}$ is the number of nonzero Fourier coefficients of $f$. This matches the previous lower bound $Q_E(f(x,y)) = \Omega(\log \rank(M_f))$ by Buhrman and de Wolf \cite{BdW01} for low-degree XOR functions. %This implies that for XOR functions $F(x,y) = f(x\oplus y)$, $\frac{1}{2}\log\rank(M_F) \leq \qcc(F) = O(2^d \log\rank(M_F))$. 
	
	%We show that for any Boolean function $f$, the bounded-error quantum communication complexity $Q(f(x+y))$ of XOR functions $f(x+y)$ satisfies that $Q_\epsilon(f(x+y)) = O(2^d (\log\|\hat f\|_{1,\epsilon} +\log (n/\epsilon))$, where $d = \deg_2(f)$ is the F2-degree of $f$, and $\|\hat f\|_{1,\epsilon} = \min_{g:\|f-g\|_\infty \leq \epsilon} \|\hat g\|_{1}$. This implies that the previous lower bound $Q(f(x+y)) = \Omega(\log\|\hat f\|_{1,\epsilon})$ by Lee and Shraibman is tight for $f$ with low F2-degree. The result also confirms the quantum version of the Log-rank Conjecture for low-degree XOR functions. In addition, we show that the exact quantum communication complexity satisfies $Q_E(f(x+y)) = O(2^d \log\|\hat f\|_0)$, where $\|\hat f\|_0$ is the number of nonzero Fourier coefficients of $f$. This matches the previous lower bound $Q_E(f(x,y)) = \Omega(\log rank(M_f))$ by Buhrman and de Wolf for low-degree XOR functions.
\end{abstract}

\section{Introduction}
Communication complexity studies the minimum amount of communication needed for a computational task with input distributed to two (or more) parties. Communication complexity has been applied to prove impossibility results for problems in a surprisingly wide range of computational models. At the heart of studies of communication complexity are lower bounds, and the tightness of lower bound techniques has been among the most important, and at the same time, most challenging questions. %In a standard setting, \alice and \bob each get an input $x$ and $y$, respectively, and they aim to compute a function $f(x,y)$ by sending messages to each other back and forth. The messages can be deterministic, randomized or quantum, resulting in the corresponding communication complexities $\D(f)$, $\R_\epsilon(f)$ and $\Q_\epsilon(f)$, where $\epsilon$ is the error probability allowed in randomized and quantum protocols. 
Indeed, one of the most famous open problems in communication complexity is the Log-rank Conjecture:
It has been known that the (two-party, interactive) deterministic communication complexity $\D(f) \geq \log_2(\rank(M_f))$ \cite{MS82}, where the rank is over $\mbR$ and $M_f$ is the communication matrix defined as $M_f(x,y) = f(x,y)$. The Log-rank Conjecture, proposed by Lov{\'a}sz and Saks \cite{LS88}, says that the above bound is polynomially tight, namely 
\begin{equation}
\D(f) = O(\log_2^{O(1)}(\rank(M_f))).
\label{eq:}
\end{equation}
A quantum version of the conjecture, also seemingly hard to attack, says that the $\epsilon$-bounded error quantum communication complexity 
\begin{equation}
	\Q_\epsilon(f) = O(\log_2^{O(1)}(\rank_\epsilon(M_f))),
\label{eq:qlogrank}
\end{equation}
where $\rank_\epsilon(M_f) = \min\{\rank_\epsilon(M_f): \|f-g\|_\infty \leq \epsilon\}$ \cite{BdW01,LS09}. Note that proving this type of conjectures needs to design efficient communication protocols.

Communication complexity for the class of XOR functions has recently drawn an increasing amount of attention \cite{ZS09,ZS10,LZ10,MO10,LLZ11,SW12,LZ13,TWXZ13}. The class contains those functions $F(x,y) = f(x\oplus y)$ for some function $f:\BntB$, where the inner operator $\oplus$ is the bit-wise XOR. Denote such functions $F$ by $f\circ \oplus$. This class includes important functions such as Equality (deciding whether $x=y$) \cite{Yao79,NS96,Amb96,BK97,BCWdW01}), Hamming Distance (deciding whether $|x\oplus y| \leq d$) \cite{Yao03,GKdW04,HSZZ06,ZS09,LLZ11,LZ13}, and Gap Hamming Distance (distinguishing $|x\oplus y| \leq n/2-\sqrt{n}$ and $|x\oplus y| \geq n/2+\sqrt{n}$) \cite{JKS08,CR12,She12,Vid12}. Communication complexity of XOR functions also exhibits interesting connections to Fourier analysis of Boolean functions. First, the rank of the communication matrix $M_f$ is nothing but $\fn{f}{0}$, the number of nonzero Fourier coefficients of $f$. Thus for XOR functions, the Log-rank Conjecture becomes the assertion that $\D(f\circ \oplus) = O(\log^{O(1)}\fn{f}{0})$; see \cite{ZS09,MO10,KS13,TWXZ13} for some investigations on this topic. The quantum Log-rank Conjecture becomes $\Q_\epsilon(f) = O(\log_2^{O(1)}(\fn{f}{0,\epsilon}))$ accordingly, where $\fn{f}{0,\epsilon} = \min\{\fn{g}{0}: \|f-g\|_\infty \leq \epsilon\}$. Second, as shown in \cite{LS09}, the quantum communication complexity for computing $f(x\oplus y)$ is known to be lower bounded by an approximate version of the Fourier $\ell_1$-norm as follows.
\begin{equation}
	\Q_\epsilon(f\circ \oplus) = \Omega(\log \fn{f}{1,\epsilon})\text{, where }\fn{f}{1,\epsilon} = \min\{\fn{g}{1}: \|f-g\|_\infty \leq \epsilon\}.
\label{eq:qlb}
\end{equation}
The tightness of this lower bound has been an intriguing question. 

In this paper, we show that the bound in Eq.\eqref{eq:qlb} is tight for functions $f$ with low \gf-degree, the degree of $f$ viewed as a polynomial in $\gf[x_1, ..., x_n]$. For convenience of comparison, we copy the lower bound in Eq.\eqref{eq:qlb} into the following theorem.
\newcommand\ThmBoundedError
{
For any function $f:\BntB$ with $\deg_2(f) = d$, and any $\epsilon \in (0,1/2^{d+4})$, we have 
	\[
		\Omega(\log\fn{f}{1,\epsilon}) \leq \Q_\epsilon(f\circ \oplus) \leq O\Big(2^d \big(\log\fn{f}{1,\epsilon} +\log \frac{n}{\epsilon}\big)\log(1/\epsilon)\Big).
	\]
}
\begin{Thm}\label{thm:bounded-error}
	\ThmBoundedError
\end{Thm}
\vspace{.5em}
This theorem has two implications on the Log-rank Conjectures. First, it was known that the above lower bound was smaller than $\log\rank_\epsilon(M_{f\circ \oplus})$, and the above upper bound satisfies $\fn{f}{1,\epsilon} \leq \fn{f}{0,\epsilon}$. Thus the above upper bound confirms the quantum Log-rank Conjecture (Eq.\eqref{eq:qlogrank}) for low-degree XOR functions.
\begin{Cor}
	The quantum Log-rank Conjecture holds for XOR functions $f$ with \gf-degree at most $O(\log\log\fn{f}{1,\epsilon})$.
\end{Cor}

Second, we also have a variant of the protocol in Theorem \ref{thm:bounded-error}, and the variant is an \emph{exact} protocol in the sense that it has a fixed number of qubits exchanged besides that it has zero error. (For comparison, the classical zero-error protocols usually refer to Las Vegas ones in which the number of bits exchanged is a random variable that can be very large, and the complexity cost measure is the expectation of the number of communication bits.) For exact protocols, we have the following theorem, where the lower bound is from \cite{BdW01}; we copy it into the following theorem, again for the convenience of comparison. 
\newcommand\ThmZeroError
{
	For any function $f:\BntB$ with $\deg_2(f) = d$, we have 
	\[
		\frac{1}{2}\log_2\fn{f}{0} \leq \Q_E(f\circ \oplus) \leq 2^{d+1}\log_2\fn{f}{0}.
	\]
	In particular, $\Q_E(f\circ \oplus)$ is polynomially related to $\log\rank(M_{f\circ \oplus})$ when $\deg_2(f) = O(\log\log\fn{f}{0})$.
}
\begin{Thm}\label{thm:zero-error}
	\ThmZeroError
\end{Thm}

In \cite{TWXZ13}, it shows that $\D(f\circ \oplus) \leq O(2^{d^2/2}\log^{d-2}\fn{f}{0})$, which implies that the Log-rank Conjecture holds for constant-degree XOR functions. The complexity bound in Theorem \ref{thm:zero-error} has a better dependence on $d$, which enables us to obtain an upper bound of $\log^{O(1)}\fn{f}{0}$ for a larger range of $d$. Another desirable property of the protocol in Theorem \ref{thm:zero-error} is that unlike the ones in \cite{TWXZ13} and most other upper bounds in communication complexity, this protocol is efficient not only in communication but also in computation, provided that the Fourier spectrum can be efficiently encoded and decoded.

\paragraph{Techniques} 
One common idea the protocols in this paper share with the ones in \cite{TWXZ13} (as well as many work in additive combinatorics) is degree reduction: The protocols have $d$ rounds and each round $i$ reduces the problem of computation of a function $f_i$ to that of another function $f_{i+1}$, with $\deg_2(f_{i+1}) \le \deg_2(f_i)-1$. Different than the protocols in \cite{TWXZ13}, the protocol in this paper are not derived from parity decision tree algorithms. Neither do they use linear polynomial rank or analyze any effect of linear restrictions on the Fourier domain as in \cite{TWXZ13}. Instead, the protocols in this paper merely use the definition of quantum Fourier transform over the additive group of $\gfn$, and the efficiency of the protocols comes directly from the Fourier sparsity of the corresponding function. Some new difficulty appears in this quantum Fourier sampling approach: $f_{i+1}$ is actually known only to \alice but not to \bob. This is solved by observing a simple (yet important) property of the collection of derivatives of $f_i$ along all directions.

%In general, quantum parity decision tree algorithms cannot be simulated efficiently by quantum protocols, because the space needed for one query is $n$. Classically this is not a problem since one query corresponds to one linear function $\ell$ on $x+y$, which equals to $\ell(x)+\ell(y)$, thus exchanging $\ell(x)$ and $\ell(y)$ using 2 bits allows \alice and \bob simulate one parity query. For the quantum parity decision tree algorithms, one query is a superposition of linear functions, thus 

\section{Preliminaries} 
Let $[n] = \{1, 2, ..., n\}$. For a vector $v\in \mbR^N$, its support is $\supp(v) = \{i\in [N]: v_i \neq 0\}$. For two $n$-bit strings $x$ and $y$, their addition, denoted $x\oplus y$ (or sometimes just $x+y$), is bit-wise over \gf. For a set $A\subseteq \Bn$, define $A+A = \{a_1 + a_2: a_1, a_2\in A\}$. In general, define $kA = \{a_1 + \cdots + a_k: a_i\in A, \forall i\in [k]\}$. It is easy to see that $|kA| \leq |A|^k$.

A Boolean function $f:\BntB$ can be viewed as a multi-linear polynomial over \gf, whose degree is called \gf-degree and denoted by $\deg_2(f)$. For a function $f:\BntB$ and a direction vector $t\in \Bn-\{0^n\}$, the derivative $\Delta_t f$ is defined by $\Delta_t f(x) = f(x) + f(x+t)$, where both additions are over \gf. It is easy to check that $\deg_2(\Delta_t f) < \deg_2(f)$ for any non-constant $f$ and any $t\in \Bn-\{0^n\}$. If one represents the range of a Boolean function by $\pmB$, the derivative becomes $\Delta_t f(x) = f(x) f(x+t)$.

For a real function $f:\Bn\to\mbR$, one can define its Fourier coefficients by $\hat f(\alpha) = 2^{-n}\sum_x f(x)\chi_\alpha(x)$, where the characters $\chi_\alpha(x) = (-1)^{\alpha\cdot x}$ are orthogonal with respect to the inner product $\langle f_1,f_2\rangle = 2^{-n}\sum_x[f_1(x)f_2(x)]$. The function $f$ can be written as $f = \sum_\alpha \hat f(\alpha) \chi_\alpha$. The Fourier sparsity of $f$, denoted by $\|\hat f\|_0$, is the number of nonzero Fourier coefficients of $f$. For any $p>0$, the $\ell_p$-norm of $\hat f$,  denoted by $\fn{f}{p}$, is $(\sum_\alpha |\hat f(\alpha)|^p)^{1/p}$. In particular, $\fn{f}{1} = \sum_\alpha |\hat f(\alpha)|$. One can also define an approximate version of the Fourier $\ell_1$-norm by $\fn{f}{1,\epsilon} = \min\{\fn{g}{1}: \|f-g\|_\infty \leq \epsilon\}$ where $\|f-g\|_\infty = \max_x |f(x) - g(x)|$. Similarly define $\fn{f}{0,\epsilon} = \min\{\fn{g}{0}: \|f-g\|_\infty \leq \epsilon\}$.

For any function $f:\Bn\to\mbR$, Parseval's Indentity says that $\sum_\alpha \hat f_\alpha^2 = \av_{x}[f(x)^2]$. When the range of $f$ is $\pmB$, this becomes $\sum_\alpha \hat f_\alpha^2 = \av_x [f(x)^2] = 1$. 

The quantum Fourier transform on $\gfn$ is defined by $\sum_x c_x \ket{x} \mapsto 2^{-n/2} \sum_{x,\alpha} c_x \chi_\alpha(x) \ket{\alpha}$, and it is easily seen to be a unitary operator. The transform can be implemented by $H^{\otimes n}$ where $H = \frac{1}{\sqrt{2}}\begin{pmatrix}1 & 1 \\ 1 & -1 \end{pmatrix}$ is the Hadamard matrix.

The following lemma extends Chernoff's bound to general domains; see, for example, \cite{DP12} (Problem 1.19).
\begin{Lem}\label{lem:Chernoff}
Suppose we have random variables $X_i \in [a_i , b_i]$ for $i=1, 2, ..., n$, and let $X = \sum_{i=1}^n X_i$. Then 
\[
	\pr[|X -\av[X]| > t] < 2e^{-\frac{2t^2}{\sum_i (b_i-a_i)^2}}.
\]
In particular, if each $X_i$ takes values in $[-1,1]$, then 
\[
	\pr[|X -\av[X]| > t] < 2e^{-\frac{t^2}{4n}}.
\]
\end{Lem}

\section{Protocol}
In this section, we will show Theorem \ref{thm:bounded-error}. We will first mention how to convert Fourier $\ell_1$-norm to Fourier $\ell_0$-norm in Section \ref{sec:l1tol0}, and then show the main protocol in Section \ref{sec:protocol}.

\subsection{From $\ell_1$-norm approximation to $\ell_0$-norm approximation}\label{sec:l1tol0}
In \cite{BS91,Gro97}, the sampling of characters with probability proportional to Fourier coefficients (in abstract value) is studied. Given a function $f:\Bn\to\mbR$, we sample $\alpha\in \Bn$ with probability $|\hat f(\alpha)| / \fn{f}{1}$. We refer to a sample from this process as a \emph{Fourier $\ell_1$-sample}. Using Lemma \ref{lem:Chernoff}, it is not hard to show the following lemma. 
\begin{Lem}[Grolmusz, \cite{Gro97}]\label{lem:l1tol0}
	For a function $g:\Bn\to\mbR$, independently draw $M = O(\fn{g}{1}^2n \log(1/\lambda)/\delta^2)$ Fourier $\ell_1$-samples $\alpha^1$, ..., $\alpha^M$. Let $h(x) = \frac{\fn{g}{1}}{M}\sum_{i=1}^M \sign(\hat g(\alpha^i))\chi_{\alpha^i}(x) $. Then  
	\[ \pr[\forall x\in \Bn, |h(x)-g(x)| \leq \delta ] \geq 1-\lambda. \]
\end{Lem}
The original lemma actually considers the probability of $\sign(h(x)) = \sign(g(x))$, but the same proof works for the above statement. We include a proof here for completeness.
\begin{proof}
	Let $Z_i = \sign(\hat g(\alpha^i))\chi_{\alpha^i}(x)\in \pmB$. Note that \[\av[Z_i] = \sum_{\alpha\in\Bn} |\hat g(\alpha)|\sign(\hat g(\alpha))\chi_{\alpha}(x)/\fn{g}{1} = g(x)/\fn{g}{1}.\] So by Lemma \ref{lem:Chernoff}, we have
	\begin{align*}
		\pr[\exists x, |h(x)-g(x)| > \delta ] \le 2^n \pr\Big[\Big|\sum_i Z_i - \frac{g(x)M}{\fn{g}{1}}\Big| > \frac{\delta M}{\fn{g}{1}}\Big] \leq 2^{n+1}e^{-\frac{\delta^2 M}{4\fn{g}{1}^2}} \le \lambda
	\end{align*}
\end{proof}

\subsection{Protocol}\label{sec:protocol}
Now we describe the protocol in this section. The setup is as follows. Suppose that there is a function $f:\Bn\to\pmB$, which can be approximated by a Fourier sparse function $g:\Bn\to\mbR$ satisfying that $\|f-g\|_\infty \leq \epsilon$. The Fourier expansion of $g$ is $g = \sum_\alpha \hat g(\alpha) \chi_\alpha$ and let $A = \supp(\hat g)$. In addition, let $d = \deg_2(f)$ and $N = 2^n$. For each $k \in \{0,1,...,d-1\}$, \alice and \bob fix an encoding $E_k:\Bn \to [|A|^{2^k}]$ \st for any $\alpha,\beta\in 2^k A$, $E_k(\alpha)\neq E_k(\beta)$. Finally, for a real function $h:\Bn\to\mbR$ define $\Delta_t h(x) = h(x)h(x+t)$. The algorithm is in Box \textbf{Algorithm 1}.

\begin{algorithm}
	\caption{Protocol \qxor for $f(x,y)$}
	\mbox {{\bf Input}: $x$ to \alice, $y$ to \bob } \\
	\mbox {{\bf Output}: $\ans\in \pmB$.} \\
	\mbox {{\bf Registers}: $C$ is a 1-qubit register and $M$ is an $n$-qubit regiester.} \\
	\mbox {{\bf Assumption}: $f$ has an approximation $g$ with $\|f-g\|_\infty \leq \epsilon$ and $\supp(\hat g) = A$.} \\
  \begin{algorithmic}[1]
		\STATE For each $k \in \{0,1,...,d-1\}$, \alice and \bob fix an encoding $E_k:\Bn \to [|A|^{2^k}]$ \st for any $\alpha,\beta\in 2^k A$, $E_k(\alpha)\neq E_k(\beta)$.
		\STATE $k := 0$, $\ans := 1$; $f^{(k)} = f$, $g^{(k)} = g$.
		\WHILE {$\deg_2(f^{(k)}) \geq 1$}
			\STATE \label{step:Alice phase} \alice creates the state 		
			\begin{equation}
			\ket{\psi} = \frac{1}{\sqrt{2}}\big(\ket{0}_C \ket{0}_M + \ket{1}_C\sum_{\alpha\in \Bn}\frac{\widehat{g^{(k)}} (\alpha)}{\wfn{g^{(k)}}{2}} \chi_\alpha(x)\ket{E_k(\alpha)}_M\big)
			\label{eq:starting-state}
			\end{equation} 
			and sends register $C$ and the last $\min\{n,\lceil 2^{k}\log |A|\rceil\} $ qubits of register $M$ to \bob.
		
			\STATE \label{step:Bob phase} \bob applies the following unitary transform: 
				\[\text{on } \ket{1}_C, \text{ apply } \ket{E_k(\alpha)}_M \to \chi_\alpha(y)\ket{E_k(\alpha)}_M,\] 
			and sends the resulting state $\ket{\psi'}$ back to \alice.
			\STATE \label{step:recover} \alice applies the following unitary transform: 
			\[\text{on } \ket{1}_C, \text{ apply } \ket{E_k(\alpha)}_M \to \ket{\alpha}_M.\]
			\vspace{-1em}
			\STATE \label{step:FT} \alice applies the quantum Fourier transform on register $M$. 
			\STATE \label{step:measureM} \alice measures register $M$ in the computational basis and observes an outcome $t\in \Bn$. 
			\STATE \label{step:measureC} \alice measures register $C$ in $\{\ket{+}, \ket{-}\}$ basis and observes an outcome $b\in \pmB$. 
			\STATE \label{step:answer} $\ans := b\cdot \ans$.
			\IF  {$t = 0$} 
				\STATE \label{step:update} \alice outputs $\ans$ and terminates the whole protocol,
			\ELSE 
				\STATE $f^{(k+1)} := \Delta_t f^{(k)}$, $g^{(k+1)} := \Delta_t g^{(k)}$, $k:=k+1$,
				\STATE \alice sends $\deg_2(f^{(k)})$ to \bob.
			\ENDIF
		\ENDWHILE
		%\IF {$\deg_2(f^{(k)}) = 1$}
		%	\STATE \label{step:final-send} \quad \bob sends $f^{(k)}(y)$ to \alice, who then outputs $\ans \cdot f^{(k)}(x) \cdot f^{(k)}(y)$ and terminates the program.
		%\ELSE %($\deg_2(f^{(k)}) = 0$)
			\STATE \alice outputs $\ans \cdot f^{(k)}(0)$ and terminates the program.
		%\ENDIF
	\end{algorithmic}
\end{algorithm}

\begin{Lem}\label{lem:error-of-der}
	For any function $f:\Bn\to\pmB$ and $g:\Bn\to\mbR$, if $\|f-g\|_\infty \leq \epsilon$, then for any $t_1$, ..., $t_k\in\Bn$, $\|\Delta_{t_1}\cdots\Delta_{t_k} f - \Delta_{t_1}\cdots\Delta_{t_k} g\|_\infty \leq (1+\epsilon)^{2^k} - 1$.
\end{Lem}

\begin{proof}
	When taking derivative once, the approximation error increases as follows. 
	\[|\Delta_{t}f(x) - \Delta_{t}g(x)| = |f(x)f(x+t) - g(x)g(x+t)| \leq (1+\epsilon)^2 - 1 = 2\epsilon + \epsilon^2.\]
	Using an induction, we can easily see that taking $k$ derivatives has the following effect on the accuracy.
	\[\|\Delta_{t_1}\cdots\Delta_{t_k}f - \Delta_{t_1}\cdots\Delta_{t_k}g\|_\infty = (1+\epsilon)^{2^k} - 1.\]
\end{proof}

\begin{Lem}\label{lem:main}
	Suppose that $f:\Bn\to\pmB$ and $g:\Bn\to\mbR$ has $\|f-g\|_\infty \leq \epsilon < 2^{-d-1}$, where $d = \deg_2(f)$. Then Protocol \qxor computes $f(x+y)$ by at most $2^{d+2}\log\fn{g}{0}$ qubits of communication, and the error probability is at most $2^{d}\epsilon$.
\end{Lem}
\begin{proof}
	Let us analyze the protocol step by step. (For the convenience of understanding, first think of $k = 0$ in the following.) %Step (\ref{step:get g}) gives a function $g$ satisfying that $\fn{g}{0} \leq M = 4\fn{f}{1,\epsilon}^2 \log(1/\lambda) / \delta^2$. In addition, for any fixed $z$, by Lemma \ref{lem:l1tol0}, \[\pr[|g(z)-f(z)| > \epsilon+\delta ] < \lambda\] . 
	In Step (\ref{step:Alice phase}), it is easy to see that the $\ell_2$-norm of the state is $\frac{1}{2} + \frac{1}{2} \sum_\alpha|\widehat{g^{(k)}}(\alpha)|^2 / \wfn{g^{(k)}}{2}^2 = 1$, thus the state in Eq.\eqref{eq:starting-state} is indeed a quantum pure state. After Step (\ref{step:Bob phase}), the state is
	\[
	\ket{\psi'} = \frac{1}{\sqrt{2}} \Big(\ket{0}_C \ket{0}_M + \ket{1}_C \sum_{\alpha\in \Bn} \frac{\widehat{g^{(k)}}(\alpha)}{\wfn{g^{(k)}}{2}}  \chi_\alpha(x+y)\ket{E_k(\alpha)}_M
	\Big).
	\]

	After decoding $\alpha$ in Step (\ref{step:recover}) and applying the quantum Fourier transform in Step (\ref{step:FT}), \alice holds the state
	\begin{align*}
		\ket{\psi''} & = \frac{1}{\sqrt{2N}}\left(\ket{0}_C \sum_{t\in \Bn}\ket{t}_M + \ket{1}_C\sum_{\alpha\in \Bn,t\in \Bn}\frac{\widehat{g^{(k)}}(\alpha)}{\wfn{g^{(k)}}{2}} \chi_\alpha(x+y)\chi_\alpha(t)\ket{t}_M\right) \\
		& = \frac{1}{\sqrt{N}}\sum_{t\in \Bn}\frac{1}{\sqrt{2}}\left(\ket{0}_C + \sum_{\alpha\in \Bn}\frac{\widehat{g^{(k)}}(\alpha)}{\wfn{g^{(k)}}{2}}  \chi_\alpha(x+y+t) \ket{1}_C\right)\ket{t}_M \\
		& = \frac{1}{\sqrt{N}}\sum_{t\in \Bn}\frac{1}{\sqrt{2}}\left(\ket{0}_C + \frac{g^{(k)}(x+y+t)}{\wfn{g^{(k)}}{2}}\ket{1}_C\right)\ket{t}_M.
	\end{align*}
	
	After the measurement in Step (\ref{step:measureM}), \alice obtains a random direction $t$, and the state left in register $C$ is $\frac{1}{\sqrt{2}}\big(\ket{0}_C + \frac{g^{(k)}(x+y+t)}{\wfn{g^{(k)}}{2}}\ket{1}_C\big)$. Then in the next step, measuring register $C$ in the $\pmB$ basis gives $f^{(k)}(x+y+t)$ with high probability. Indeed, by Parseval's Identity, 
	\begin{equation}\label{eq:norm of gk}
		\wfn{g^{(k)}}{2} = \|g^{(k)}\|_2 = \sqrt{\av_x [g^{(k)}(x)^2]} \le (1+\epsilon)^{2^k}.
	\end{equation}
	Thus when \alice measures $C$, she observes $\frac{1}{\sqrt{2}}(\ket{0}+f^{(k)}(x+y+t)\ket{1})$ with probability 
	\begin{equation}
		\left(\frac{1}{2} + \frac{f^{(k)}(x+y+t)g^{(k)}(x+y+t)}{2\wfn{g^{(k)}}{2}}\right)^2 \geq \left(\frac{1}{2} + \frac{1-((1+\epsilon)^{2^k} - 1)}{2(1+\epsilon)^{2^{k}}}\right)^2 = (1+\epsilon)^{-2^k} \ge 1-2^k\epsilon,
	\label{eq:error-in-one-round}
	\end{equation}
	where the first inequality uses Lemma \ref{lem:error-of-der} and Eq.\eqref{eq:norm of gk}.

	Now we explain Step (\ref{step:answer}) and (\ref{step:update}). If $t$ happens to be 0, then \alice already gets $g^{(k)}(x+y)$ which well approximates $f^{(k)}(x+y)$. In general $t\neq 0$. Turning around the definition of derivative, $\Delta_t f(x+y) = f(x+y)f(x+y+t)$, we have that $f(x+y) = f(x+y+t) \Delta_t f(x+y)$. Since we have obtained $f(x+y+t)$, the problem of computing $f(x+y)$ reduces to that of computing $\Delta_t f$ on the same input $x+y$. This reduction is implemented in Step (\ref{step:answer}),  and we let $f^{(k+1)} = \Delta_t f^{(k)}$ and go to the next iteration. Therefore, each round reduces the problem to computing the derivative, which is a lower degree polynomial, on the same input. Finally, when the degree of the polynomial is 0, the function is constant, thus \alice can easily compute it as the last line after the \textbf{while} loop in the algorithm.
	
	%Two issues arise in this approach. First, the Fourier sparsity of the derivative increases up to quadratically, and the  
	One issue in this approach is that only \alice knows $t$ after Step \ref{step:measureM}, but \bob does not know $t$ and consequently does not know $f^{(k+1)} = \Delta_t f^{(k)}$ for the next round. Also note that it is unaffordable for \alice to send the whole $t$ to \bob. Therefore, it seems hard for \alice and \bob to coordinate on $E_k$. The solution here is to note that for all $h:\Bn\to\mbR$, and \emph{for all} $t\in\Bn$, we have 
	\[\supp(\widehat{\Delta_t h}) \subseteq \supp(\hat h) + \supp(\hat h).\] 
	Indeed, denote $h_t(x) = h(x+t)$, then \[\widehat{h_t}(\alpha) = \av_x [h(x+t)\chi_\alpha(x)] = \av_x [h(x)\chi_\alpha(x)\chi_\alpha(t)] = \chi_\alpha(t)\hat h(\alpha).\] Therefore, 
	\begin{align}
		\widehat{\Delta_t h}(\alpha) = \widehat{h \cdot h_t}(\alpha) = \sum_\beta \hat h(\beta+\alpha)\hat h_t(\beta) = \sum_\beta \hat h(\beta+\alpha)\hat h(\beta)\chi_t(\beta).
	\end{align}
	If $\alpha\notin \supp(\hat h) + \supp(\hat h)$, then there is simply no $\beta$ \st both $\hat h(\beta)$ and $\hat h(\beta+\alpha)$ are nonzero. This implies that $\supp(\widehat{\Delta_t h}) \subseteq \supp(\hat h) + \supp(\hat h)$. Using the same argument, it is easily seen that in general, for any $t_1, ..., t_{k}\in \Bn$, the derivative $g^{(k)} = \Delta_{t_1}\cdots \Delta_{t_{k}}g$ has Fourier support contained in $2^{k}A$. Observe that the only operation \bob makes in each round $k$ is to add a phase $\chi_\alpha(y)$ on $\ket{E_k(\alpha)}$. So \alice and \bob can fix an encoding $E_k:\Bn \to [|A|^{2^k}]$ \st $E_k(\alpha)\neq E_k(\beta)$ for any $\alpha,\beta\in 2^k A$. \footnote{It is admittedly true that for a particular set of directions $t_1, ..., t_k\in \Bn$, the Fourier spectrum for $g^{(k)} = \Delta_{t_1}\cdots \Delta_{t_{k}}g$ is only a subset of $2^k A$, thus $\hat g^{(k)} (\alpha) = 0$ for some $\alpha \in 2^k A$. But this does not affect the correctness of the protocol, though some communication is wasted in coping with \bob's ignorance of $t$.} Since for any $t_1, ..., t_{k}\in \Bn$, $E_k(\supp(g^{(k)})) \subseteq [|A|^{2^k}]$, the encoding $E_k$ is injective on $\supp(g^{(k)})$, and thus \alice and \bob can decode in Steps (\ref{step:Bob phase}) and (\ref{step:recover}). This also explains why \alice only needs to send the last $\min\{n,\lceil 2^{k}\log |A|\rceil\} $ qubits of register $M$ to \bob in Step (\ref{step:Alice phase}).
	
	Next we analyze the error probability. The protocol is correct as long as in each iteration $k$, the observed outcome $b$ in Step (\ref{step:measureC}) is equal to $f^{(k)}(x+y)$. Since each iteration $k$ has error probability $2^k\epsilon$ as showed in Eq.\eqref{eq:error-in-one-round}, applying the union bound over $k=1, 2, ..., d-1$ gives that the probability that there exists one round $k$ in which the output bit disagrees with $\Delta_{t_1,...,t_{k-1}}f(x+y+t_k)$ is at most 
	$\sum_{k=0}^{d-1} 2^k\epsilon \leq  2^{d}\epsilon.$
	
	Finally we analyze the communication cost. In the \textbf{while} loop, only Step (\ref{step:Alice phase}) and (\ref{step:Bob phase}) need communication of $1+\lceil 2^{k}\log |A|\rceil$ qubits each. Since taking derivative decreases the \gf-degree by at least 1, we know that $k\leq \deg_2(f)-1$ before the \textbf{while} loop ends. %Considering the extra bit sent in Step (\ref{step:final-send}), 
	The total communication cost is at most 
	\[\sum_{k=0}^{d-1} 2(1 + \lceil 2^{k}\log |A|\rceil) \le 2^{d+1}\log|A| + 2d < 2^{d+2} \log|A|\] 
	qubits. 
\end{proof}

Now we are ready to prove Theorem \ref{thm:bounded-error}.
\begin{Thm-bounded-error}[Restated]
	\ThmBoundedError
\end{Thm-bounded-error}
\begin{proof}
	The lower bound is from \cite{LS09}. For the upper bound, by definition, there is a function $g:\BntR$ with $\|f-g\|_\infty \leq \epsilon$ and $\fn{g}{1} = \fn{f}{1,\epsilon}$. We first use Lemma \ref{lem:l1tol0} to get a function $h$ with $\|f-h\|_\infty \leq 2\epsilon$ and $\fn{h}{0} \leq O(\fn{g}{1}^2 n/\epsilon^2)$. Then we use the protocol \qxor and Lemma \ref{lem:main} to obtain a protocol of error probability $2^{d+2}\epsilon \le 1/4$ and communication cost $O(2^d\log \fn{h}{0}) = O(2^d (\log\fn{g}{1} +\log \frac{n}{\epsilon}))$. Repeat the protocol for $k = O(\log(1/\epsilon))$ times to reduce the error probability to $\epsilon$, and the communication cost is 
	\[O\Big(2^d \big(\log\fn{g}{1} +\log \frac{n}{\epsilon}\big)\log(1/\epsilon)\Big) = O\Big(2^d \big(\log\fn{f}{1,\epsilon} +\log \frac{n}{\epsilon}\big)\log(1/\epsilon)\Big).\]
\end{proof}

Given the above proof, Theorem \ref{thm:zero-error} is an easy corollary.
\begin{Thm-zero-error}[Restated]
	\ThmZeroError
\end{Thm-zero-error}
\begin{proof}
	The lower bound is from \cite{BdW01}. The upper bound is from Lemma \ref{lem:main} by letting $g = f$.
\end{proof}	

Finally we notice that all the steps, except for the encoding and decoding of $E_k$, can be implemented efficiently.
\begin{Prop}
	The protocol in Theorem \ref{thm:zero-error} needs only $O(dn)$ Hadamard gates, C-NOT gates and single-qubit measurements, $2^d$ calls of $f$, plus the computation for encoding and decoding of $\{E_k\}$.
\end{Prop}
\begin{proof}
	The quantum Fourier transform on $\Bn$ can be implemented by $n$ Hadamard gates, and all other steps except for the encoding and decoding can also be implemented using $O(n)$ CNOT gates and single-qubit measurements. The only step that may need explanation is when \alice prepares the initial state
	\[
	\ket{\psi} = \frac{1}{\sqrt{2}}\big(\ket{0}_C \ket{0}_M + \ket{1}_C\sum_{\alpha\in \Bn}\widehat{f^{(k)}} (\alpha) \chi_\alpha(x)\ket{E_k(\alpha)}_M\big)
	\]
	This can indeed by implemented easily as follows. \alice prepares $\ket{+}_C\ket{0}_M$, and conditioned on $C$ being $\ket{1}$, applies quantum Fourier transform on $M$ to get $\frac{1}{\sqrt{N}}\sum_z \ket{z}_M$. Now \alice adds the phase $f^{(k)}(z)$ on $\ket{z}$ by $2^k$ calls to $f$. After applying the quantum Fourier transform again on $M$, \alice obtains the state $\frac{1}{N} \sum_\alpha \sum_z f^{(k)}(z) \chi_\alpha(z) \ket{\alpha}_M = \sum_\alpha \widehat{f^{(k)}}(\alpha) \ket{\alpha}_M$. Then \alice adds the phase $\chi_\alpha(x)$ on $\ket{\alpha}_M$ and gets $\sum_\alpha \widehat{f^{(k)}}(\alpha) \chi_\alpha(x) \ket{\alpha}_M$. Finally \alice encodes $\alpha$ and gets the state $\sum_\alpha\widehat{f^{(k)}}(\alpha) \chi_\alpha(x) \ket{E_k(\alpha)}_M$, as desired.
\end{proof}

%\section{Next}
%Generalize to functions $F:\mbF_q^n\times \mbF_q^n\to \mbF_q$ defined by $F(x,y) = f(x+y)$ for some $f\in \gfq[x_1, ..., x_n]$?
%
%Generalize to sign matrices $A$ and prove that $\Q(f) = O(2^d \log \|A\|_{tr,\epsilon})$? 
%
%Give an example to show quadratic speedup of $Q_E$ and $R_\epsilon$ for a total function? Maybe starting from degree-4 polynomials?
%
\vspace{1em}
\subsection*{Acknowledgments}
The author would like to thank Ronald de Wolf and Zhaohui Wei for valuable comments on earlier version of the paper. Part of the research was conducted when the author visited Tsinghua University in China (China Basic Research Grant 2011CBA00300, sub-project 2011CBA00301) and Centre of Quantum Technologies in Singapore partially 	under their support. The research was also supported by Research Grants Council of the Hong Kong S.A.R. (Project no. CUHK418710, CUHK419011). %Part of the research was done when 

\bibliography{logrank}

\end{document}